\documentclass[11pt]{article}
\usepackage[utf8]{inputenc}
\usepackage[left=1in, top=1in, right=1in, bottom=1in]{geometry}
\usepackage{soul}
\usepackage{amsthm}
\usepackage{amsmath}
\usepackage{amssymb}
\usepackage{xcolor}
\usepackage{setspace}
\usepackage{booktabs} 
\usepackage{algorithm}
\usepackage{algorithmic}
\usepackage[colorlinks=true, citecolor=blue, linkcolor=red, urlcolor=black]{hyperref}
\usepackage[switch]{lineno}
\usepackage[nocomma, short]{optidef}

\DeclareMathOperator*{\argmax}{arg\,max}
\DeclareMathOperator*{\argmin}{arg\,min}

\usepackage{amsfonts}
\usepackage{mathtools}
\usepackage{bm}
\usepackage{natbib}
\usepackage{mdframed}
\mdfsetup{
 skipabove=4pt,
 skipbelow=4pt,
 leftmargin=0pt,
 rightmargin=0pt,
 innertopmargin=3pt,
 innerbottommargin=3pt,
 innerleftmargin=3pt,
 innerrightmargin=3pt,
 frametitleaboveskip=3pt,
 frametitlebelowskip=3pt,
 frametitlerule=true,
 splittopskip=2\topsep}
\usepackage{enumitem}

\newtheorem{definition}{Definition}
\newtheorem{theorem}{Theorem}
\newtheorem{lemma}{Lemma}
\newtheorem{proposition}{Proposition}
\newtheorem{observation}{Observation}

\DeclareMathOperator*{\E}{\mathbf{E}}
\newcommand{\dd}{\mathrm{d}}
\newcommand{\I}{\mathbb{I}}

\ttfamily

\begin{document}
\title{Optimal Fixed-Price Mechanism with Signaling}

\author{
Zhikang Fan \\ Renmin University of China\\ \texttt{fanzhikang@ruc.edu.cn}
\and
Weiran Shen\\ Renmin University of China \\ \texttt{shenweiran@ruc.edu.cn}
}

\maketitle

\begin{abstract}
Consider a trade market with one seller and multiple buyers. The seller aims to sell an indivisible item and maximize their revenue. This paper focuses on a simple and popular mechanism--the fixed-price mechanism. Unlike the standard setting, we assume there is information asymmetry between buyers and the seller. Specifically, we allow the seller to design information before setting the fixed price, which implies that we study the mechanism design problem in a broader space. We call this mechanism space the fixed-price signaling mechanism.

We assume that buyers' valuation of the item depends on the quality of the item. The seller can privately observe the item's quality, whereas buyers only see its distribution. In this case, the seller can influence buyers' valuations by strategically disclosing information about the item's quality, thereby adjusting the fixed price. We consider two types of buyers with different levels of rationality: ex-post individual rational (IR) and ex-interim individual rational. We show that when the market has only one buyer, the optimal revenue generated by the fixed-price signaling mechanism is identical to that of the fixed-price mechanism, regardless of the level of rationality. Furthermore, when there are multiple buyers in the market and all of them are ex-post IR, we show that there is no fixed-price mechanism that is obedient for all buyers. However, if all buyers are ex-interim IR, we show that the optimal fixed-price signaling mechanism will generate more revenue for the seller than the fixed-price mechanism.
\end{abstract}
\setcounter{tocdepth}{1} 

\onehalfspacing

\newpage 

\section{Introduction}
\label{sec:introduction}

The phenomenon of information asymmetry is ubiquitous in the real world and has garnered extensive research attention in both computer science and economics, including security~\citep{rabinovich2015information,xu2015exploring}, advertising~\citep{badanidiyuru2018targeting,emek2014signaling}, transportation~\citep{wu2021value,gan2022bayesian,fan2024fast} and voting~\citep{castiglioni2020persuading}. In these applications, an agent with more information, referred to as the sender, can influence the behavior of other agents by strategically disclosing information, which is also known as ``Bayesian persuasion''~\citep{kamenica2011bayesian} or information design~\citep{bergemann2016information}. In these information asymmetric scenarios, the sender is often in a leading position and can obtain higher utility than when he has no additional information.

There is also a body of work studying the optimal auction design when the seller can disclose information~\citep{esHo2007optimal,chen2020signalling,wei2022price}, or under specific auction formats such as second-price auction~\citep{bergemann2022optimal} and posted price auction~\citep{castiglioni2022signaling}. In this paper, we introduce information asymmetry into a trade market and limit the design space to a simple format--the fixed-price mechanism. We aim to study the impact of enabling the seller to design information on the optimal mechanism and seller revenue within the fixed-price mechanism space, and answer the following questions:
\begin{enumerate}
    \item What is the optimal mechanism for the seller when they can design information?
    \item Can allow the seller to design information lead to higher revenue? 
\end{enumerate}

Specifically, we consider a trade market with one seller and multiple buyers. The seller aims to sell an indivisible item. The quality of the item can only be observed by the seller, whereas buyers can only see its distribution. Each buyer's valuation of the item depends on its quality. Given the information advantage regarding the item's quality, the seller can design information before determining the fixed price. \cite{castiglioni2022signaling} also investigated a similar situation, except that the signal sent by the seller is correlated with the proposed price and buyers arrive sequentially. 

\subsection{Our Contributions}
In this paper, we consider two types of behavior patterns of buyers: the ex-post individual rational buyers and the ex-interim individual rational buyers. We refer to the mechanism space in which the seller can design information within a fixed-price mechanism as the fixed-price signaling mechanism. Therefore, throughout this paper, we focus on two types of design spaces: the fixed-price mechanism and the fixed-price signaling mechanism.

\paragraph{Optimal fixed-price mechanism} Before delving into our main problems, we first discuss the optimal fixed-price mechanism for the seller under these two behavior patterns. For ex-post IR buyers, the optimal fixed price is a trade-off between the item's price and the probability that the item will be sold. However, for ex-interim IR buyers, the seller can continuously raise the item's price as long as at least one buyer is willing to buy it.

\paragraph{Optimal fixed-price signaling mechanism for single buyer} As a warm-up, we consider a special case of our original setup and study the optimal fixed-price signaling mechanism when there is only one buyer in the market. Surprisingly, we find that the revenue generated from the optimal fixed-price signaling mechanism is identical to that of the fixed-price mechanism, regardless of the behavior patterns employed by the buyer. This implies that in this case, allowing the seller to design information will not bring him more revenue, and the advantage in information does not translate into an advantage in terms of revenue.

\paragraph{Optimal fixed-price signaling mechanism for multiple buyers} Furthermore, we investigate the optimal fixed-price signaling mechanism when there are multiple buyers in the market. For ex-post IR buyers, we find that there is no fixed-price signaling mechanism that is obedient for all the buyers. However, if we assume that buyers cannot buy the item when the seller does not recommend it, we show that this problem can be solved in closed form. As for ex-interim IR buyers, we show that the optimal fixed-price signaling mechanism will bring more revenue for the seller than the fixed-price mechanism.

\subsection{Related Work}
\paragraph{Information design} Our research is grounded in the literature on information design. We adopt the ``Bayesian persuasion'' framework, proposed by the seminal work~\citep{kamenica2011bayesian}, to model how the seller designs information. Like most follow-up work, we assume there is only one side that has private information. One of the related works in this line is the model proposed by~\citet{Castiglioni2021SignalingIB}, where an informed sender persuades a set of uninformed receivers. However, there is a substantial distinction between persuasion and selling an item through persuasion. In our context, the seller has only one item for sale, so even if many buyers are willing to buy the item, only one buyer will eventually get the item.

\paragraph{Joint design mechanism and information} Our work aligns with the research on the combination of mechanism design with information design \citep{esHo2007optimal,wei2022price, castiglioni2022signaling, FanS23, fan2024optimal}. \citet{esHo2007optimal} also consider a setting where an item seller sells one indivisible item to multiple buyers. However, in their model, the seller cannot observe the item's quality. Another difference lies in the mechanism space, they do not restrict design space while we focus on the fixed-price mechanism. \citet{wei2022price} study a single buyer setting under a general design space and give a closed-form solution. In contrast, we consider there are multiple buyers in the market. Closer to us are the series of works that focus on information design under specific auction formats, such as second-price auction~\citep{bergemann2022optimal} and posted price auction~\citep{castiglioni2022signaling}. \citet{castiglioni2022signaling} consider a posted price auction where buyers arrive sequentially and their valuations for the item depend on a random state that is only observed by the seller. However, the key difference between their model and ours is that in their model, the seller's price function depends on the signal sent by the seller but we focus on a constant price. \citet{chen2020signalling} also adopt a type-dependent pricing strategy for the seller in a setting with one buyer and a binary type space.

\paragraph{\bf Structure of the paper}
The rest of the paper is organized as follows. Section \ref{sec:preliminary} describes the model and two types of mechanism space considered in this paper. Section \ref{sec:fixed-price} investigates the optimal fixed-price mechanism without signaling. In section \ref{sec:fixed-price signaling}, we first study the optimal fixed-price signaling mechanism design problem under one buyer setting and later generalize it to the multi-buyer setting. We summarize all the results in Section \ref{sec:summary}. Finally, we conclude in Section \ref{sec:conclusion}.

\section{Preliminaries}
\label{sec:preliminary}
\subsection{Model}
\label{sec:model}
Consider a trade market with one seller and multiple buyers. The seller has an indivisible item for sale and buyers want to buy it. The item has a quality, denoted by $q\in Q$, which can only be observed by the seller. We assume that $q$ is a random variable drawn from publicly known distribution $G(q)$. The support of $q$ is $Q = [q_1, q_2 ]$. $G(q)$ is differentiable in its support, with the corresponding probability density function $g(q)$.

All buyers possess no private information. Hence their valuation of the item only depends on its quality $q$. Let $N = \{1, \dots, n \}$ denote the set of buyers and $v_i: Q \mapsto \mathbb{R}^+$ be the valuation function of buyer $i$. The higher the quality of the item, the higher the buyer's valuation of the item should be. Thus, we also assume that $v_i(q)$ is monotone increasing with respect to $q$ for all $i$. Subsequently, we can define the inverse function of $v_i(\cdot)$ as $v_i^{-1}(\cdot)$. All buyers observe the public distribution $G(q)$ and derive a prior valuation for the item. Formally, the prior valuation of buyer $i$ can be denoted by $\E_{q\sim G(q)} [v_i(q)]$.

\paragraph{Rationality} We consider two types of buyers: the ex-post rational buyer and the ex-interim rational buyer. Below, we define these two types of buyers.

\begin{definition}[Ex-post rational buyer]
An ex-post rational buyer will purchase an item if and only if his valuation of the item is not less than $p$ after knowing the actual $q$. Formally, it is equivalent to:
\begin{align*}
    v_i(q) \ge p.
\end{align*}
\end{definition}
We can also define the ex-post utility of the buyer as follows:
\begin{align*}
    U_{post}(\pi, p) = v_i(q) -p.
\end{align*}

\begin{definition}[Ex-interim rational buyer]
An ex-interim rational buyer will purchase an item if and only if his expected valuation of the item is not less than $p$. Formally, it is equivalent to:
\begin{align*}
    \E_{q} [v_i(q)] \ge p.
\end{align*}
\end{definition}
Similarly, we define the ex-interim utility of the buyer as follows:
\begin{align*}
    U_{interim}(\pi, p) = \E_{q} [v_i(q)] -p.
\end{align*}

\subsection{Mechanism Space}
From the seller's perspective, we aim to design a revenue-maximizing item-selling mechanism for the seller. Next, we describe the set of mechanisms considered in this paper. Thus, the optimal mechanism is the one that generates the highest revenue within the mechanism space.

Firstly, we require that the price of the item be identical for all buyers, that is, there is no price discrimination. Second, we allow the sender to design information before deciding on the fixed price. According to whether the seller can design information, we consider two types of mechanisms separately: the fixed-price mechanism and the fixed-price signaling mechanism.

\subsubsection*{Fixed-price Mechanism}
The fixed-price mechanism is a straightforward yet widely used mechanism~\citep{roesler2017buyer,zhang2021fixed,liu2023improved} where the seller sets a single, constant price for an item that applies to all buyers. We denote this fixed price as $p$. When the item is sold, the buyer who obtains the item pays $p$ to the seller.

In a fixed-price mechanism, the interaction between the seller and buyers takes place as follows:
\begin{enumerate}
    \item The seller observes the quality of the item $q$, while the buyers observe the distribution $G(q)$.
    \item The seller sets a constant price $p$ for the item and makes it known to all buyers.
    \item If at least one buyer is willing to buy the item\footnote{If multiple buyers are willing to buy the item, the seller selects one among them randomly.}, then the item is sold and the seller receives $p$. Otherwise, the seller retains the item.
\end{enumerate}

\subsubsection*{Fixed-price Mechanism with Signaling }
As mentioned in step (1) of the fixed price mechanism, the seller has more information about the item's quality than buyers. According to the well-known ``Bayesian Persuasion'' framework \citep{kamenica2011bayesian}, the seller can expand the design space by sending signals to buyers before setting prices, thereby influencing buyers' behavior and potentially obtaining higher revenue. We refer to this new design space as the fixed-price signaling mechanism.

\paragraph{Information design} Following the ``Bayesian Persuasion'' framework, the seller can disclose information by way of signaling. Specifically, the seller first commits to a signaling scheme, which is a mapping from the quality set to a distribution over a signal set. Then, after observing the item's quality, the seller will send a signal to each buyer based on the committed signaling scheme. After receiving the signal, each buyer will update their belief over $q$ based on the Bayes update rule.

Next, we formally describe the fixed-price signaling mechanisms.

\begin{definition}[Fixed-price signaling mechanism]
A fixed-price signaling mechanism $\mathcal{M}$ can be described by a tuple $(\pi, p)$, where:
\begin{itemize}
    \item $\pi: Q \mapsto \Delta(\Sigma)$ is the signaling scheme and $\Sigma$ is a signal set. When the item's quality is $q$, the seller will send signal $\sigma \in \Sigma$ with probability $\pi(q, \sigma)$.
    \item $p$ is a constant price for the item. 
\end{itemize}
\end{definition}

In a fixed-price signaling mechanism $\mathcal{M} = (\pi, p)$, the interaction between the seller and buyers occurs as follows:
\begin{enumerate}
    \item The seller observes the item's quality $q$, and the buyers observe the distribution $G(q)$.
    \item The seller sends signal $\sigma \in \Sigma$ drawn from distribution $\pi(q, \cdot)$.
    \item After receiving the signal, buyers update their beliefs and decide whether to buy.
    \item If at least one buyer is willing to buy, then the item is sold and the seller receives $p$.
\end{enumerate}

\section{Fixed-Price Mechanism Without Signaling}
\label{sec:fixed-price}
In this section, we discuss the optimal fixed-price mechanism for the seller.

\subsection{Ex-post Rational Buyers}
For ex-post rational buyers, the higher the price of an item, the lower the probability that a buyer will buy it. Hence, the seller needs to strike a balance between the probability of the item being sold and its price when designing the fixed-price mechanism.

\begin{proposition}
When buyers are ex-post rational, the optimal fixed price $p^*$ should be set as follows:
\begin{align*}
    p^* \in \argmax_p Rev_{fix}(p),
\end{align*}
where $Rev_{fix}(p) = [1- \prod_{i\in N} G(v_i^{-1}(p))] \cdot p$ denotes the expected revenue from setting a fixed price $p$.

\end{proposition}

\begin{proof}
For ex-post rational buyers, they will only be willing to buy the item if the valuation is greater than or equal to $p$ after knowing the quality $q$. In fixed-price mechanisms, buyers have a common prior belief $G(q)$. Thus, if the item's price is $p$, the probability that buyer $i$'s valuation of the item is less than $p$ is:
\begin{align*}
    Pr\{ v_i(q) < p \} = Pr\{ q < v_i^{-1}(p) \} = G(v_i^{-1}(p)).
\end{align*}
Then the probability that all buyers' valuation is less than $p$ is $\prod_{i\in N} G(v_i^{-1}(p))$, which also represents the probability that the item cannot be sold. So the probability that at least one buyer is willing to buy the item is $1 - \prod_{i\in N} G(v_i^{-1}(p))$.

Based on the above analysis, the seller's expected revenue from setting a fixed price $p$ can be written as:
\begin{align*}
    Rev_{fix}(p) = [1- \prod_{i\in N} G(v_i^{-1}(p))] \cdot p.
\end{align*}
Then the seller can optimally set the price as follows:
\begin{align*}
    p^* \in \argmax_p Rev_{fix}(p).
\end{align*}
This concludes the proof.
\end{proof}

\subsection{Ex-interim Rational Buyers}
In fixed-price mechanisms, since there is no additional information about $q$, the distribution $G(q)$ remains constant. The buyer's valuation of the item, which is related to $G(q)$, is also a constant. Consequently, the seller can increase the price as much as they can while ensuring that at least one buyer will buy the item.

\begin{proposition}
When buyers are ex-interim rational, the optimal fixed price $p^*$ should be set as:
\begin{align*}
    p^* = \max_i \bar{v}_i,
\end{align*}
where $\bar{v}_i = \E_{q\sim G(q)}[v_i(q)]$ is buyer $i$'s expected valuation of the item.
\end{proposition}

\begin{proof}
For ex-interim rational buyers, they will only be willing to buy the item if their expected valuation of the item is greater than or equal to $p$. Note that in fixed-price mechanisms, the valuation of buyer $i$ for the item is a constant:
\begin{align*}
     \bar{v}_i =  \E_{q\sim G(q)} [v_i(q)].
\end{align*}
Thus the buyer $i$ will not be willing to buy the item if the fixed price $p > \bar{v}_i$. Moreover, there is no buyer will buy the item if $p > \max_i \bar{v}_i$. 

Denoted by $ p_0 = \max_i \bar{v}_i$, which is the highest price at which at least one buyer is willing to purchase the item. Then the seller's revenue from setting a fixed price $p$ can be written as:
\begin{align*}
    Rev_{fix} (p) = \I \{ p \leq p_0 \} \cdot p,
\end{align*}
where $\I\{ \cdot\}$ is an indicator function. It is apparent that to maximize revenue, the fixed price should be set at $p_0$, that is:
\begin{align*}
     p^* = \argmax_p Rev_{fix}(p) = p_0.
\end{align*}
This concludes the proof.
\end{proof}

\section{Fixed-Price Mechanism with Signaling}
\label{sec:fixed-price signaling}
In the previous section, we discussed the optimal fixed-price mechanism. In this section, we consider a broader space--the fixed-price signaling mechanism in which the seller is allowed to design information before determining the price of the item.

\subsection{Warm-up: Single Buyer in the Market}
In this section, as a warm-up, we consider a special case of our original model. We assume there is only one buyer in the market, so the seller only needs to design information for one player. In this scenario, each signal $\sigma \in \Sigma$ only needs to be one-dimensional.

Upon receiving signal $\sigma$, the buyer will update his belief about $q$ and decide whether to buy the item. According to the revelation principle, it is without loss of generality to regard each signal as an action recommendation, since each signal will induce a posterior belief that leads to a certain action~\cite{kamenica2011bayesian,dughmi2016algorithmic}. In our setup, there are two actions for the buyer: buy or not buy. Hence, we only need two signals in set $\Sigma$. We say a mechanism is obedient if the buyer will always follow the action recommendation.

\begin{definition}[Obedience]
A mechanism $(\pi, p)$ is obedient if the buyer has no incentive to deviate from the action recommendation sent by the seller.
\end{definition}

Let $\Sigma = \{0, 1 \}$, where signal $1$ corresponds to ``buy'' and signal $0$ corresponds to ``not buy''. For simplicity, we use $\pi(q)$ to denote the probability of sending signal $1$ when the item's quality is $q$. Naturally, $1-\pi(q)$ denotes the probability of sending signal $0$. 

Upon receiving signal 1, the buyer will update his belief over $q$ as follows:
\begin{align}
\label{eq:update 1}
    g(q|1) = \frac{ \pi(q) \cdot g(q) }{ \int_{q' \in Q} \pi(q') \cdot g(q') \,\dd q' }.
\end{align}
Similarly, after receiving signal 0, the posterior belief of the buyer over $q$ is:
\begin{align}
\label{eq:update 0}
    g(q|0) = \frac{ [1-\pi(q) ] \cdot g(q) \,\dd q }{ \int_{q' \in Q } [1- \pi(q')] \cdot g(q') \,\dd q'  } .
\end{align}

Given an obedient mechanism $(\pi, p)$, the seller's revenue can be written as:
\begin{align}
\label{eq:rev sig}
    Rev_{sig} (\pi, p) = \int_{q\in Q} \pi(q) g(q) \,\dd q \cdot p.
\end{align}

Next, we discuss what constraints an obedient mechanism should satisfy and what the optimal mechanism is when facing an ex-post rational buyer and an ex-interim rational buyer, respectively.

\subsubsection{Ex-post rational buyer}
We first discuss what constraints should an obedience mechanism satisfy when facing an ex-post rational buyer.

Specifically, to ensure the buyer's obedience, two constraints need to be imposed on the mechanism $(\pi, p)$: (1) After receiving signal 1, the buyer's ex-post utility from purchasing the item should be at least 0; (2) After receiving signal 0, the buyer's ex-post utility from purchasing the item should be at most 0.

When receiving signal 1, the buyer obtains posterior belief $g(q|1)$ based on Equation \eqref{eq:update 1}. Given this, we derive the probability that the valuation of the buyer is less than $p$ as follows:
\begin{align*}
    Pr\{ v(q) < p \mid 1 \} =& Pr\{ q <v^{-1} (p) \mid 1 \}\\
    =& \int_{q_1}^{v^{-1}(p)} g(q|1) \,\dd q\\
    =& \frac{\int_{q_1}^{v^{-1}(p)} \pi(q) \cdot g(q) \,\dd q }{\int_{q'} \pi(q') \cdot g(q') \,\dd q'}.
\end{align*}

Then the probability that the buyer is willing to buy the item is $1- Pr\{ v(q) < p \mid 1 \}$. Now we only need to ensure that after receiving signal 1, the probability $1-Pr\{ v(q) < p \mid 1 \} = 1$, that is:
\begin{align}
\label{eq:ob signal 1}
    \int_{q_1}^{v^{-1}(p)} \pi(q) \cdot g(q) \,\dd q = 0.
\end{align}
Similarly, after receiving signal 0, the buyer derives posterior belief $g(q|0)$ based on Equation \eqref{eq:update 0}. Given this, the probability that the buyer's valuation is less than $p$ is:
\begin{align*}
    Pr\{ v(q) < p \mid 0 \} &= \int_{q_1}^{v^{-1}(p)} g(q|0) \,\dd q \nonumber \\
    &= \frac{\int_{q_1}^{v^{-1}(p) } [1-\pi(q)] \cdot g(q) \,\dd q  }{\int_{q'\in Q}  [1-\pi(q')] \cdot g(q') \,\dd q' } .
\end{align*}
Then we need to ensure that after receiving signal 0, the above probability is equal to 1, which is equivalent to:
\begin{align}
\label{eq:ob signal 0}
    \int_{q_1}^{v^{-1}(p) } [1-\pi(q)] g(q) \,\dd q = \int_{Q}  [1-\pi(q)] g(q) \,\dd q .
\end{align}

\begin{theorem}
When the buyer is ex-post rational, the following signaling scheme $\pi^*$ and fixed price $p^*$ forms an optimal fixed-price signaling mechanism:
\begin{align*}
    \pi^*(q) = \begin{cases}
        0 & \text{ if } q < v^{-1}(p) \\
        1 & \text{ otherwise } 
    \end{cases},
\end{align*}
\begin{align*}
    p^* =  \argmax_p [1-G(v^{-1}(p))] \cdot p.
\end{align*}
\end{theorem}

\begin{proof}
According to obedience constraint \eqref{eq:ob signal 1}, we should set $\pi(q)$ to zero when $q$ is between $q_1$ and $v^{-1}(p)$, that is:
\begin{align*}
    \pi(q) = 0  \text{ if } q_1 \le q < v^{-1}(p).
\end{align*}
And we can rewrite obedience constraint \eqref{eq:ob signal 0} as follows:
\begin{align*}
    \int_{v^{-1}(p)}^{q_2} [1-\pi(q)] g(q)\,\dd q = 0.
\end{align*}
Thus we should set $\pi(q) = 1$ when $q$ is between $v^{-1}(p)$ and $q_2$. Overall, we obtain the optimal signaling scheme as follows:
\begin{align*}
    \pi^*(q) = \begin{cases}
        0 & \text{ if } q_1 \le q < v^{-1}(p) \\
        1 & \text{ if } v^{-1}(p) \le q \le q_2
    \end{cases}.
\end{align*}

 Given the signaling scheme $\pi^*$, the seller's revenue from setting the fixed price as $p$ can be written as:
\begin{align*}
    Rev_{sig}(\pi^*, p) =& Pr\{ \pi^*(q) = 1 \} \cdot p \\
    =& Pr\{ q\ge v^{-1}(p) \} \cdot p \\
    =&[1-G(v^{-1}(p))] \cdot p.
\end{align*}
So the optimal fixed price should be set as follows:
\begin{align*}
    p^* \in \argmax_p Rev_{sig}(\pi^*, p).
\end{align*}
This concludes the proof.
\end{proof}

Surprisingly, the optimal revenue obtained from the fixed-price signaling mechanism is identical to that obtained within the fixed-price mechanism space. 
\begin{proposition}
\label{prop: ex-post IR}
When facing an ex-post rational buyer, the optimal revenue obtained in the fixed-price signaling mechanism is identical to that obtained in the fixed-price mechanism.
\end{proposition}

This implies that allowing the seller to design information does not result in higher returns for him.

\subsubsection{Ex-interim rational buyer}
Recall that without additional information, the buyer's expected valuation of the item is a constant. However, when the seller can design information, they can influence the buyer's expected valuation through signaling. Next, we discuss what constraints an obedient mechanism should satisfy when facing an ex-interim rational buyer.

Given posterior belief $g(q|1)$, the ex-interim valuation of the buyer for the item is equivalent to:
\begin{align*}
    \E_{q\sim g(q|1)} [v(q)] = \int_{q\in Q} v(q) g(q|1)  \,\dd q.
\end{align*}
So to ensure the buyer is willing to buy the item after receiving signal 1, we need to require that the ex-interim utility of the buyer is no less than 0, that is:
\begin{align}
\label{eq:ob interim original}
    \int_{q\in Q} g(q|1) v(q) \,\dd q - p  \ge 0.
\end{align}
With some simple algebraic manipulations, we obtain:
\begin{align}
\label{eq:ob interim signal 1}
    \int_{q\in Q} \pi(q) [v(q) - p] g(q) \,\dd q \ge 0.
\end{align}
Similarly, after receiving signal 0, the valuation of the buyer is:
\begin{align*}
    \E_{q \sim g(q|0)} [v(q)] = \int_{q\in Q} v(q) g(q| 0) \,\dd q.
\end{align*}
To ensure the buyer will not buy the item after receiving signal 0, we need to require that:
\begin{align*}
    \int_{q\in Q}  v(q) g(q|0) \,\dd q - p \le 0.
\end{align*}
With some simple algebraic manipulations, we obtain:
\begin{align}
\label{eq:ob post signal 0}
    \int_{q\in Q} \pi(q) [v(q) - p] g(q) \,\dd q \ge \E_{q \sim g(q)} [v(q)] - p.
\end{align}

Combine with the seller's objective, we can formulate the optimal mechanism design problem as the following optimization program:
\begin{maxi}
{\pi, p}
{ Rev_{sig} (\pi, p) = \int_{q\in Q} \pi(q) g(q)\,\dd q \cdot p }
{\label{eq:LP}}
{}
\addConstraint{ \int_{q\in Q} \pi(q) [v(q) - p] g(q) \,\dd q }{\ge 0 }{}
\addConstraint{ \int_{q\in Q} \pi(q) [v(q) - p] g(q) \,\dd q }{\ge \E_{q\sim g(q) } [v(q)] - p }{}
\end{maxi}

From the obedience constraint, we can obtain an upper bound of the above optimization problem. 
\begin{proposition}
The optimization program \eqref{eq:LP} is upper bounded by:
\begin{align*}
    Rev_{sig}(\pi, p) \le \E_{q} [v(q)].
\end{align*}
\end{proposition}

\begin{proof}
According to constraint \eqref{eq:ob interim signal 1}, we have:
\begin{align*}
    \int_{q\in Q} \pi(q) g(q) \,\dd q \cdot p &\le  \int_{q\in Q} \pi(q) v(q) g(q) \,\dd q\\
    & \le \int_{q\in Q} v(q) g(q) \, \dd q\\
    & = \E_{q\sim g(q)} [v(q)].
\end{align*}
This concludes the proof.
\end{proof}

Next, we show that we can construct an obedience mechanism that can achieve this upper bound, thus achieving optimal.

\begin{theorem}
\label{the:post one}
When the buyer is ex-interim rational, the following signaling scheme $\pi^*$ and fixed price $p^*$ forms an optimal fixed-price signaling mechanism:
\begin{align*}
    \pi^*(q) = 1, \forall q\in Q \quad \text{ and } \quad p^* = \E_{q} [v(q)].
\end{align*}
\end{theorem}

\begin{proof}
The proof process is decomposed into two parts. First, we ignore constraint \eqref{eq:ob post signal 0} and focus on the remaining optimization program. Second, we show that the optimal solution in the relaxed problem also satisfies constraint \eqref{eq:ob post signal 0}, thus remaining optimal in the original program \eqref{eq:LP}.

Firstly, we increase the price $p$ such that the obedience constraint \eqref{eq:ob interim signal 1} is binding, that is:
\begin{align*}
    \int_{q\in Q} \pi(q) v(q) g(q) \,\dd q = \int_{q\in Q} \pi(q) g(q) \,\dd q \cdot p .
\end{align*}
Put the price $p$ on one side, we get:
\begin{align}
\label{eq:price mid process}
    p = \frac{\int_{q\in Q} \pi(q) v(q) g(q) \,\dd q}{\int_{q\in Q} \pi(q) g(q) \,\dd q}.
\end{align}
Constraint \eqref{eq:ob interim signal 1} ensures that the buyer will buy the item after receiving signal 1. Hence, at this price, the buyer will buy the item after receiving signal 1. Therefore, to maximize revenue, the seller should send signal 1 as frequently as possible. We construct the following signaling scheme:
\begin{align*}
    \pi^*(q) = 1, \forall q\in Q,
\end{align*}
which means regardless of the value of $q$, the seller sends signal 1 and the buyer will always buy the item. Then Equation \eqref{eq:price mid process} becomes:
\begin{align*}
    p^* = \int_{q\in Q} v(q) g(q) \,\dd q = \E_q [v(q)] .
\end{align*}

Next, we show the constructed mechanism $(\pi^*, p^*)$ also satisfies the omitted constraint \eqref{eq:ob post signal 0}. Putting $(\pi^*, p^*)$ into constraint \eqref{eq:ob post signal 0}, we obtain:
\begin{align*}
    \int_{q\in Q} [v(q) - p] g(q) \,\dd q &= \int_{q\in Q} v(q) g(q) \,\dd q - p = \E_{q} [v(q)] - p.
\end{align*}

This concludes the proof.
\end{proof}

In hindsight, we can observe that since $\pi= 1$ for all $q\in Q$, the probability of sending signal 0 is 0. This is the reason why when solving the program \eqref{eq:LP}, we can safely ignore constraint \eqref{eq:ob post signal 0}.

Based on Theorem \ref{the:post one}, we have the following observation.
\begin{observation}
When facing an ex-interim rational buyer, the optimal revenue achieved in the fixed-price signaling mechanism is identical to that obtained in the fixed-price mechanism.
\end{observation}

Combined with Observation \eqref{prop: ex-post IR}, we can draw the conclusion that when in the market with one buyer, allowing the seller to design information will not bring him more revenue.
\begin{proposition}
In a market with one buyer, allowing the seller to design information before setting the fixed price will not bring him more revenue.
\end{proposition}

\subsection{Multiple buyers in the Market}
In this section, we consider the original model as described in Section \ref{sec:model}. In this model, there are $n$ buyers in the market and the seller has just one item for sale.

In this scenario, the signal set $\Sigma$ should be $n$-dimension. The seller can send different signals to different buyers, resulting in different posterior beliefs. The signal space will be extremely large, but with the help of the revelation principle, we can reduce the signal space.

\begin{lemma}[\citet{bergemann2018design}]
It is without loss of generality to focus on the responsive experiment where the signal space has at most the cardinality of the outcome space.
\end{lemma}

Based on the above results, we can focus on the set of signaling schemes where a one-to-one correspondence exists between signals and outcomes. In our setting, there are $n+1$ possible outcomes, with $n$ of them corresponding to each buyer obtaining the item and an additional one corresponding to no buyer buying the item. Thus, we can define $\Sigma$ as follows:
\begin{align*}
    \Sigma =  \left\{ \bm{\sigma } \in \{0, 1\}^n : \sum_{i=1}^n \sigma_i \le 1   \right\},
\end{align*}
where $\bm{\sigma}$ with $\sigma_i = 1$ corresponds to the outcome where buyer $i$ obtains the item, and with $\sigma_i=0$ for all $i$ corresponds to the outcome where no buyer makes a purchase. From an implementation perspective, the seller can send the $i$-th element of $\bm{\sigma }$ to the buyer $i$, indicating whether the buyer should make a purchase or not. For simplicity, we denote the signal with $\sigma_i=1$ as $s_i$, and the signal with $\sigma_i=0$ for all $i$ as $s_0$. Therefore, sending signal $s_i$ means the seller asks buyer $i$ to buy the item.

There is only one item for sale, so $\pi$ should satisfy the following constraints:
\begin{align}
\label{eq:prob constraint}
    \sum_{i \in N} \pi(q, s_i) + \pi(q, s_0) = 1 \quad \text{and} \quad \pi(q, s_i) \ge 0, \forall i,  \forall  q.
\end{align}
Upon receiving signal $1$, buyer $i$ will update belief over $q$ as follows:
\begin{align}
\label{eq:belief update multi 1}
    g(q|1) = \frac{\pi(q, s_i) \cdot g(q)}{\int_{q'\in Q} \pi(q', s_i) \cdot g(q') \,\dd q' }.
\end{align}
Similarly, upon receiving signal $0$, the posterior belief of buyer $i$ is:
\begin{align}
    g(q|0) = \frac{  [1-\pi(q, s_i)] \cdot g(q)  }{ \int_{q'\in Q} [1- \pi(q', s_i)] \cdot g(q') \,\dd q' }.
\end{align}

Whether there are multiple buyers or a single buyer, the seller's revenue is equal to the probability of the item being sold multiplied by the fixed price, as shown in Equation \eqref{eq:rev sig}.

\subsubsection{Ex-post rational buyers} We first discuss what constraints an obedience mechanism should satisfy when there are multiple ex-post rational buyers. 

Similar to the analysis in the single-buyer case, upon receiving signal 1, the probability that buyer $i$'s valuation is less than $p$ is:
\begin{align*}
    Pr \{ v_i(q) < p \mid 1 \} = \frac{\int_{q_1}^{v_i^{-1}(p)} \pi(q, s_i) \cdot g(q) \,\dd q}{ \int_{q'\in Q} \pi(q', s_i) \cdot g(q') \,\dd q'  }.
\end{align*}
To ensure buyer $i$ will buy the item after receiving signal 1, we need to ensure that $Pr\{ v_i(q) < p \mid 1 \} = 0$, that is:
\begin{align*}
    \int_{q_1}^{v_i^{-1}(p)} \pi(q, s_i) g(q) \,\dd q = 0.
\end{align*}
Similarly, after receiving signal 0, the probability that buyer $i$'s valuation is less than $p$ is:
\begin{align*}
    Pr \{ v_i(q) < p \mid 0 \} &= \int_{q_1}^{v_i^{-1}(p)} g(q|0) \,\dd q\\
    &= \frac{\int_{q_1}^{v_i^{-1}(p)} [1- \pi(q, s_i)] \cdot g(q) \,\dd q }{\int_{q'\in Q} [1- \pi(q', s_i)] \cdot g(q') \,\dd q'  }.
\end{align*}
Then to ensure obedience, it is equivalent to $Pr\{v_i(q) < p \mid 0\} = 1$, that is:
\begin{align*}
    \int_{q_1}^{v_i^{-1}(p)} [1- \pi(q, s_i)] g(q) \,\dd q = \int_{q\in Q} [1- \pi(q, s_i)] g(q) \,\dd q,
\end{align*}
which can also be rewritten as:
\begin{align*}
    \int_{v_i^{-1}(p)}^{q_2} [1- \pi(q, s_i)] g(q) \,\dd q = 0.
\end{align*}
\begin{proposition}
When facing multiple ex-post rational buyers, a mechanism $(\pi, p)$ is obedience if and only if:
\begin{align}
    &\int_{q_1}^{v_i^{-1}(p)} \pi(q, s_i) g(q) \,\dd q = 0, \forall i\in N. \label{eq:ob post multiple 1} \\
    &\int_{v_i^{-1}(p)}^{q_2} [1- \pi(q, s_i)] g(q) \,\dd q = 0, \forall i\in N. \label{eq: ob post multiple 0}
\end{align}
\end{proposition}
We find that constraints \eqref{eq:ob post multiple 1} and \eqref{eq: ob post multiple 0} cannot be satisfied simultaneously. In other words, there is no obedient mechanism.
\begin{theorem}
When multiple ex-post rational buyers are in the market, there is no fixed-price signaling mechanism that is obedience for all the buyers.
\end{theorem}
\begin{proof}
According to the obedience constraint \eqref{eq:ob post multiple 1}, we have:
\begin{align*}
    \pi(q, s_i) = 0 \text{ if } q \in [q_1, v_i^{-1}(p) ), \forall i\in N .
\end{align*}
To satisfy constraint \eqref{eq: ob post multiple 0}, we have:
\begin{align*}
    \pi(q, s_i) = 1 \text{ if } q\in [ v_i^{-1}(p), q_2], \forall i\in N.
\end{align*}
Overall, we obtain the following signaling scheme:
\begin{align*}
    \pi(q, s_i) = \begin{cases}
        1 & \text{ if } i=0 \text{ and } q \in [q_1, v_{min}^{-1}(p) )\\
        0 & \text{ if } i\in N \text{ and } q \in [q_1, v_{min}^{-1}(p) )\\
        0 & \text{ if } i=0 \text{ and } q \in [ v_{min}^{-1}(p), q_2]\\
        1 & \text{ if } i \in N \text{ and } q \in [ v_{i}^{-1}(p), q_2]
    \end{cases},
\end{align*}
where $v_{min}^{-1}(p) = \min_{i\in N} v_{i}^{-1}(p)$.

Note that in the above signaling scheme, the seller sends $s_i$ with probability 1 for all $i\in N$ when $q\in [ v_{i}^{-1}(p), q_2] $, which contradicts the probability constraint \eqref{eq:prob constraint}. Therefore, there is no obedient fixed-price signaling mechanism.

\end{proof}

However, if we assume that buyers cannot purchase the item when the seller does not recommend it, that is, we omit constraint \eqref{eq: ob post multiple 0}, and then we can obtain the following results.

\begin{theorem}
If buyers are only allowed to buy the item after receiving signal 1, the following mechanism $(\pi^*, p^*)$ is one of the optimal ones:
\begin{align*}
    \pi^*(q, s_i) = \begin{cases}
        1 & \text{ if } i=0 \text{ and } q \in [q_1, v_{min}^{-1}(p) )\\
        0 & \text{ if } i\in N \text{ and } q \in [q_1, v_{min}^{-1}(p) )\\
        0 & \text{ if } i=0 \text{ and } q \in [v_{min}^{-1}(p), q_2 ] \\
        1 & \text{ if } i=j \text{ and } q\in [v_{min}^{-1}(p), q_2]
    \end{cases},
\end{align*}
where $j = \argmin_{i} v_{i}^{-1}(p)$.
\begin{align*}
    p^* \in \argmax_{p} [1 - G(v_{min}^{-1}(p) ) ] \cdot p.
\end{align*}

\end{theorem}

\begin{proof}
Now the obedience mechanism is fully characterized by constraint \eqref{eq:ob post multiple 1}. To satisfy this constraint, we construct the following signaling scheme:
\begin{align*}
    \pi(q, s_i) = \begin{cases}
        1 & \text{ if } i=0 \text{ and } q \in [q_1, v_{min}^{-1}(p) )\\
        0 & \text{ if } i\in N \text{ and } q \in [q_1, v_{min}^{-1}(p) )\\
        0 & \text{ if } i=0 \text{ and } q \in [v_{min}^{-1}(p), q_2 ] \\
        1 & \text{ if } i=j \text{ and } q\in [v_{min}^{-1}(p), v_{max}^{-1}(p)]
    \end{cases},
\end{align*}
where $v_{max}^{-1}(p) = \max_{i\in N} v_{i}^{-1}(p)$.

When $q\in [v_{max}^{-1}(p), q_2]$, all buyers will be willing to buy the item after receiving signal 1. Numerous obedient signaling schemes can be constructed as long as the following constraints are satisfied:
\begin{align*}
    \sum_{i\in N} \pi(q, s_i) = 1.
\end{align*}
For simplicity, we set $\pi(q, s_i) = 1$ if $i = j$ and $q \in [v_{max}^{-1}(p), q_2] $.

Under signaling scheme $\pi$, the seller's revenue is:
\begin{align*}
    Rev_{sig}(\pi, p) &= Pr\{ \sum_{i\in N} \pi(q, s_i) \} \cdot p\\
    &= Pr \{ q \in [v_{min}^{-1}(p), q_2]  \} \cdot p\\
    &= [1 - G(v_{min}^{-1}(p) ) ] \cdot p.
\end{align*}
Thus, the optimal fixed price can be obtained as follows:
\begin{align*}
    p^* \in \argmax_{p} [1 - G(v_{min}^{-1}(p) ) ] \cdot p.
\end{align*}
This concludes the proof.
\end{proof}

\subsubsection{Ex-interim rational buyers}
Next, we discuss what constraints should an obedient mechanism satisfy when facing multiple ex-interim IR buyers.

We need to ensure that the ex-interim utility of buyer $i$ from buying the item is no less than 0 when receiving signal 1. Combine Equation \eqref{eq:ob interim original} and \eqref{eq:belief update multi 1}, we get:
\begin{align}
\label{eq:ob interim multi 1}
    \int_{q\in Q} \pi(q, s_i) [v_i(q) -p] g(q) \,\dd q \ge 0, \forall i\in N.
\end{align}
Similarly, to ensure buyer $i$ will not buy the item after receiving signal 0, we have:
\begin{align}
\label{eq:ob interim multi 0}
    \int_{q\in Q} \pi(q, s_i) [v_i(q) - p] g(q) \,\dd q \ge \E_{q\sim g(q)} [v_i(q)] - p, \forall i\in N.
\end{align}

Combine with the seller's objective, we can formulate the optimal mechanism design problem as the following program:
\begin{maxi}
{\pi, p}
{ \int_{q\in Q} \sum_{i\in N} \pi(q, s_i) g(q)\,\dd q \cdot p }
{\label{eq:LP multi}}
{}
\addConstraint{ \int_{q\in Q} \pi(q, s_i) [v_i(q) -p] g(q) \,\dd q }{\ge 0 ,}{\forall i\in N}
\addConstraint{ \int_{q\in Q} \pi(q, s_i) [v_i(q) - p] g(q) \,\dd q }{\ge \E_{q} [v_i(q)] - p, }{\forall i\in N}
\end{maxi}

\begin{table}[htbp]
\centering
\caption{A summary of the optimal revenue obtained from the fixed-price mechanism and the fixed-price signaling mechanism under different buyers' behavior patterns}
\label{tab:summary}
\begin{tabular}{c|c|c}
\toprule
\textbf{Behavior pattern} &  Ex-post rational  & Ex-interim rational   \\
\midrule
\textbf{Fixed-price mechanism} & $\max_p [1- \prod_{i\in N} G(v_i^{-1}(p))] \cdot p$ & $\max_i \E[v_i(q)]$ \\
\midrule
\textbf{Fixed-price signaling mechanism(one buyer)} & $\max_p [1-G(v^{-1}(p))] \cdot p$ & $ \E [v(q)] $ \\
\midrule
\textbf{Fixed-price signaling mechanism } & $-$ , $\max_p [1 - G(v_{min}^{-1}(p) ) ] \cdot p$  & $\E [ \max_i v_i(q)  ]$  \\
\bottomrule
\end{tabular}
\end{table}

We can obtain an upper bound based on constraint \eqref{eq:ob interim multi 1}.
\begin{proposition}
The program \eqref{eq:LP multi} is upper bounded by:
\begin{align*}
    Rev_{sig}(\pi, p) \le \E_{q} [v_{max}(q)],
\end{align*}
where $v_{max}(q) = max_{i} v_i(q)$.
\end{proposition}
\begin{proof}
According to constraint \eqref{eq:ob interim multi 1}, we have:
\begin{align*}
    \int_{q\in Q} \pi(q, s_i) v_i(q) g(q) \,\dd q \ge \int_{q\in Q} \pi(q, s_i) g(q)\,\dd q \cdot p, \forall i \in N.
\end{align*}
Sum over all constraints, we obtain:
\begin{align*}
    \int_{q\in Q} \sum_{i\in N} \pi(q, s_i) g(q) \,\dd q * p &\le \int_{q\in Q} \sum_{i\in N} \pi(q, s_i) v_i(q) g(q) \,\dd q \\
    &\le \int_{q\in Q}  v_{max}(q) g(q) \,\dd q = \E_{q} [v_{max}(q) ],
\end{align*}
This concludes the proof. 
\end{proof}

Next, we demonstrate that we can construct an obedience mechanism that achieves this upper bound, thereby achieving optimality.
\begin{theorem}
When there are multiple ex-post rational buyers in the market, the following signaling $\pi^*$ and fixed price $p^*$ forms an optimal mechanism:
\begin{align*}
    \pi^*(q, s_i) = \begin{cases}
        1 & \text{ if } i \in \argmax_i v_i(q)\\
        0 & \text{ otherwise }
    \end{cases}.
\end{align*}
\begin{align*}
    p^* = \E_{q} [v_{max}(q)].
\end{align*}
\end{theorem}

\begin{proof}
Similar to the proof of Theorem \ref{the:post one}, we first ignore constraint \eqref{eq:ob interim multi 0}, and then we show that the optimal solution also satisfies constraint \eqref{eq:ob interim multi 0}, thus remaining optimal. 

First, we construct a signaling scheme as follows:
\begin{align*}
    \pi(q, s_i) = \begin{cases}
        1 & \text{ if } i \in \argmax_i v_i(q)\\
        0 & \text{ otherwise }
    \end{cases}.
\end{align*}
It means that for any $q$, the seller will send signal 1 to the buyer $i$ with the highest $v_i(q)$. Then the obedience constraint \eqref{eq:ob interim multi 1} becomes:
\begin{align*}
    \int_{q\in Q} v_{max} (q) g(q) \,\dd q \ge \int_{q\in Q} g(q) \,\dd q \cdot p = p
\end{align*}
So we can set the fixed price as follows:
\begin{align*}
    p &= \int_{q\in Q} v_{max} (q) g(q) \,\dd q\\
    &= \E_{q} [v_{max}(q)].
\end{align*}
Next, we show the constructed mechanism $(\pi^*, p^*)$ also satisfies the omitted constraint \eqref{eq:ob interim multi 0}. Plugging $(\pi^*, p^*)$ into constraint \eqref{eq:ob interim multi 0}, the right hand side of this constraint becomes:
\begin{align*}
    \E_q [v_i(q) ] - \E_{q} [v_{max}(q) ] \le 0.
\end{align*}
Therefore, this constraint is already implied by constraint \eqref{eq:ob interim multi 1}. 

This concludes the proof.
\end{proof}

\section{Summary}
\label{sec:summary}
So far we have studied the revenue-maximizing mechanism design problem for an item seller under two different buyers' behavior patterns and two design spaces. Next, we summarize all the results, as presented in Table~\ref{tab:summary}.

It should be noted that when buyers are ex-post rational and the mechanism is a fixed-price signaling mechanism, the result consists of two terms. The first term, represented by a “-”, indicates no obedient mechanism exists. The second term represents the revenue obtained under the assumption that buyers cannot purchase the item when the seller does not recommend it. From Table~\ref{tab:summary} we can draw the following conclusions:
\begin{itemize}
    \item When there is one buyer in the market, no matter which behavior patterns the buyer adopts, the revenues obtained from two mechanism spaces are the same.
    \item When there are multiple buyers in the market, and all buyers are ex-interim rational, the revenue generated by the fixed-price signaling mechanism is no less than the revenue generated by the fixed-price mechanism.
\end{itemize}
The first conclusion implies that allowing the seller to design information does not bring him more revenue. Conversely, the second conclusion indicates that when there are multiple ex-interim rational buyers in the market, the seller will obtain more revenue by designing information before determining the fixed price.

\section{Conclusion}
\label{sec:conclusion}

In this paper, we studied the optimal fixed-price signaling mechanism design problem for an item seller in a market with multiple buyers. We considered two types of buyers with different levels of rationality. As a benchmark, we first investigated the optimal fixed-price mechanism without signaling. Furthermore, we dived into our main questions and found that when there is only one buyer in the market, the revenues generated from both mechanisms are identical. However, when multiple buyers are in the market and all of them are ex-interim rational, the seller will obtain more revenue from using the fixed-price signaling mechanism.

\bibliographystyle{apalike}
\bibliography{bib}




\end{document}